\documentclass[12pt]{article}

\usepackage[a4paper, total={6.8in, 10in}]{geometry}

\RequirePackage{amsfonts, amssymb, amsthm, amsopn}
\RequirePackage{amsmath}
\RequirePackage{bm} 

\RequirePackage{pifont}
\RequirePackage{mathrsfs}
\RequirePackage[retainorgcmds]{IEEEtrantools}

\RequirePackage{fix-cm} 

\usepackage{xcolor}
\usepackage{xargs}

\usepackage[IEEEtran]{research18}
 \usepackage[caption=false,font=footnotesize]{subfig}

\begin{document}

\newcommand{\delb}{\epsilon}

\title{Chernoff-Stein-Type Exponent in Testing Between Two Outlier Distributions}

\author{Ligong Wang\thanks{The author is with the Department of Information Technology and Electrical Engineering, ETH Zurich, 8092 Zurich, Switzerland (email: ligwang@isi.ee.ethz.ch).}}

\date{}

\newcommand{\PX}{\set{P}(\set{X})}
\newcommand{\mH}{\const{H}}
\newcommand{\thecon}{\hat{\theta}_2}

\maketitle

\begin{abstract}
  Among $2^{nR}$ length-$n$ random sequences, one sequence is an outlier whose index is random. Under hypothesis $\mH_0$, the components of the outlier are independent and identically distributed (IID) according to $Q_0$, whereas under hypothesis $\mH_1$ they are IID according to $Q_1$. The remaining $(2^{nR}-1)$ sequences are mutually independent, independent of the outlier, and IID according to $P$ under both hypotheses. Based on the observation of all $2^{nR}$ sequences, one wishes to decide between $\mH_0$ and $\mH_1$. Under the constraint that the decision error probability under $\mH_0$ must be bounded away from $1$, we determine the fastest exponential decay rate of the decision error probability under $\mH_1$.
\end{abstract}

\section{Introduction}

In binary hypothesis testing, one wishes to decide between two alternative explanations for observed data. A classical setup is as follows. Under the null hypothesis $\mH_0$, the components of the random sequence $X^n$ are independent and identically distributed (IID) according to a probability mass function (PMF) $Q_0$ on finite alphabet $\set{X}$. Under the alternative hypothesis $\mH_1$, the components are IID according to another PMF $Q_1$ on $\set{X}$. Let $\phi$ denote the test function that maps the realization of $X^n$ to either $\mH_0$ or $\mH_1$. Let $\alpha$ denote the probability of a decision error under $\mH_0$, i.e., it is the probability that $\phi(X^n)=\mH_1$ when $\mH_0$ is true, and let $\beta$ denote the probability of a decision error under $\mH_1$. The Chernoff-Stein lemma \cite{coverthomas06} asserts that, if $\alpha$ must be bounded away from $1$, then the fastest achievable exponential decay rate of $\beta$ with respect to $n$ is given by the relative entropy $D(Q_0\|Q_1)$. Formally, for all $\bar{\alpha}\in(0,1)$,
\begin{equation}\label{eq:beta}
\lim_{n\to\infty} \frac{1}{n} \log \left(\min_{\alpha\le \bar{\alpha}} \beta\right) = - D(Q_0\|Q_1),
\end{equation}
where
\begin{equation}
D(Q_0\|Q_1) = \begin{cases} \displaystyle \sum_{x\in\set{X}} Q_0(x) \log \frac{Q_0(x)}{Q_1(x)}, & Q_0 \ll Q_1\\ \infty, & \textnormal{otherwise.} \end{cases}
\end{equation}
Here and throughout, logarithms are to base $2$.

The present work also studies binary hypothesis testing between some sequence $X^n$ being IID $Q_0$ and it being IID $Q_1$, but $X^n$ is now observed together with a large number of other sequences that are IID according to another PMF $P$. Specifically, there are $2^{nR}$ sequences, among which one is IID according to either $Q_0$ or $Q_1$ depending on the hypothesis, while all other sequences are IID according to $P$ under both hypotheses. It is assumed that the index of the hypothesis-dependent sequence---which we call the ``outlier''---is not known a priori. We also seek the fastest exponential decay rate of $\beta$ (error probability under $\mH_1$) when $\alpha$ (error probability under $\mH_0$) must be bounded away from $1$. Note that the setup allows $P$ to be the same as either $Q_0$ or $Q_1$, in which case the problem becomes one of deciding whether an outlier exists or not.

The main result of this work is a closed-form expression of said exponent expressed in terms of $P$, $Q_0$, $Q_1$, and $R$; see Theorem~\ref{thm:main} ahead.

The problem we study is quite different from the one that is usually called ``outlier hypothesis testing'' in the literature \cite{liniveeravalli14}. In that problem, one also observes a large number of sequences, of which most are IID according to a ``typical'' PMF, while some (possibly more than one) sequences are IID according to an ``outlier'' PMF. The goal there is usually to identify the indices of all the outliers. 
In our setup, determining the index of the outlier is not part of the task.

The present work is related to \emph{distributed hypothesis testing} \cite{ahlswedecsiszar86,han87}, in particular, to schemes there that employ \emph{random binning} \cite{shimokawahanamari94,kochmanwang25}; we shall return to this subject towards the end of the paper. 
In fact, the decision rule we propose (and hence also the proof of the direct part of the main theorem) is nearly identical to the one in \cite{kochmanwang25}. The converse proof here is new. We introduce a notion of typicality for an ensemble of sequences whose total number is exponentially large in the sequence length. We then compare the probability of a ``typical ensemble'' under both hypotheses to derive upper bounds on the probability of decision errors.

The rest of this paper is organized as follows. Section~\ref{sec:result} formally sets up the problem and presents the main result; Section~\ref{sec:typicality} briefly reviews types and typical sequences; Section~\ref{sec:typicalcoll} introduces the notion of typical ensembles; Sections~\ref{sec:direct} and~\ref{sec:converse} prove the direct and the converse parts of the main result, respectively; and Section~\ref{sec:conclusion} concludes the paper with some remarks.

\section{Setup and Result}\label{sec:result}

Let $\set{X}$ be a finite set, and let $\PX$ denote the probability simplex on $\Reals^{|\set{X}|}$. Fix three PMFs
\begin{equation}
P,Q_0,Q_1\in\PX.
\end{equation}
A ensemble of $\lceil 2^{nR} \rceil$ length-$n$ random sequences
\begin{equation}\label{eq:randomensemble}
\{ X^n(1),\ldots, X^n(\lceil 2^{nR} \rceil)\}
\end{equation}
is generated as follows. An index $K$ is chosen uniformly at random from $\{1,\ldots,\lceil 2^{nR} \rceil\}$. All components of all sequences in the ensemble except those of $X^n(K)$ are IID $\sim P$ and independent of $K$. Under hypothesis $\mH_0$, the components of $X^n(K)$ are IID $\sim Q_0$, whereas, under hypothesis $\mH_1$, they are IID $\sim Q_1$; under both hypotheses, $X^n(K)$ is independent of $K$ and of all the other sequences in the ensemble. 

A hypothesis test $\phi$ maps the realization of \eqref{eq:randomensemble} to either $\mH_0$ or $\mH_1$. We require the probability of a decision error under $\mH_0$ to be bounded away from $1$, and seek the fastest decay rate of the decision error probability under $\mH_1$. Our main result is the following.

\begin{theorem}\label{thm:main}
Let $p(n,\bar{\alpha})$ denote the smallest attainable error probability under $\mH_1$ when the decision error probability under $\mH_0$ is at most $\bar{\alpha}$, i.e.,
\begin{equation}
p(n,\bar{\alpha}) \triangleq \min_\phi \Pr \left[ \phi\big(\{ X^n(1),\ldots, X^n(\lceil 2^{nR} \rceil)\} \big) = \mH_0 \,\middle|\, \mH_1 \textnormal{ true}\right]
\end{equation}
subject to $\Pr \left[ \phi(\{ X^n(1),\ldots, X^n(\lceil 2^{nR} \rceil)\} ) = \mH_1 \,\middle|\, \mH_0 \textnormal{ true}\right] \le \bar{\alpha}$. If
\begin{equation}
P \ll Q_1,
\end{equation}
then, for all $\bar{\alpha}\in(0,1)$,
\begin{IEEEeqnarray}{rCl}
\label{eq:main}
\lim_{n\to\infty} \frac{1}{n} \log p(n,\bar{\alpha}) & = & - \min\{ \theta_1,\theta_2\}
\end{IEEEeqnarray}
where 
\begin{IEEEeqnarray}{rCl}
\theta_1 & \triangleq & D(Q_0\| Q_1)\\
\theta_2 & \triangleq & \bigl( D(Q_0\|P) - R\bigr)^+ + \min_{P'\colon D(P'\|P)\le R} D(P'\|Q_1). \label{eq:theta2}
\end{IEEEeqnarray}
\end{theorem}

\begin{IEEEproof}
The direct part is proven in Section~\ref{sec:direct} and the converse part  in Section~\ref{sec:converse}.
\end{IEEEproof}

\begin{remark}
As we shall see, the proof of the direct part of Theorem~\ref{thm:main} does not use the assumption $P\ll Q_1$. Thus, even if $P\not\ll Q_1$, we still have
\begin{equation}
\limsup_{n\to\infty} \frac{1}{n} \log p(n,\bar{\alpha}) \le  - \min\{ \theta_1,\theta_2\}.
\end{equation}
\end{remark}

\begin{remark}
If we take $R \downarrow 0$, then $\theta_2$ approaches $D(Q_0\|P) + D(P\|Q_1)$, which can be smaller than $\theta_1$, so the exponent $\min\{\theta_1,\theta_2\}$ does \emph{not} always reduce to the Chernoff-Stein exponent $\theta_1=D(Q_0\|Q_1)$. This is because, when $R \downarrow 0$, the problem we study reduces not to one where there is only one $X^n$-sequence that is IID $\sim Q_0$ or $Q_1$, but rather one where there are two sequences, one IID $\sim P$, and the other IID $\sim Q_0$ or $Q_1$. 
\end{remark}

\section{Review of Types and Typical Sequences}\label{sec:typicality}
Here we briefly review the notion of types, typical sequences, and some of their properties, and introduce some notation that we use. For more details, we refer the reader to \cite{csiszarkorner11}. 

\begin{definition}
Given a sequence $x^n\in\set{X}^n$, its {type}, or {empirical distribution}, is a PMF $\tau$ defined as follows: for every $\xi\in\set{X}$,
\begin{equation}
\tau( \xi) = \frac{1}{n} \sum_{i=1}^n \boldsymbol{1}\{ x_i = \xi\}.
\end{equation}
We sometimes denote the type of $x^n$ by $\tau[x^n]$.
\end{definition}

\begin{definition}
A sequence $x^n\in\set{X}^n$ is said to be $\epsilon$-typical with respect to distribution $P$ if its type $\tau$ satisfies, for all $\xi\in\set{X}$, 
\begin{equation}\label{eq:typical}
\bigl| \tau( \xi) - P(\xi) \bigr| \le \epsilon \cdot P(\xi).
\end{equation}
The set of all $n$-sequences that are $\epsilon$-typical with respect to $P$ is denoted $\set{T}_\epsilon^{(n)}(P)$.
\end{definition}

Note that \eqref{eq:typical} implies 
\begin{IEEEeqnarray}{rCl}
\supp(\tau)& = & \supp(P) \quad\textnormal{and}\quad
\|\tau-P\|_1  \le  \epsilon.
\end{IEEEeqnarray} 

We recall the following well-known results.

\begin{proposition}\label{prop:prob}
The probability that $X^n$ IID $\sim$P equals $x^n$ only depends on its type $\tau=\tau[x^n]$ and is given by
\begin{equation}
P^{\times n} (x^n) = 2^{-n \bigl( H(\tau) + D(\tau\| P)\bigr)}.
\end{equation}
\end{proposition}

\begin{proposition}\label{prop:LLN}
If $X^n$ is generated IID $\sim P$, then, for all $\epsilon>0$, as $n\to\infty$, the probability that $X^n$ is $\epsilon$-typical with respect to $P$ tends to one:
\begin{equation}
\lim_{n\to\infty} P^{\times n} \left[ X^n  \in\set{T}_\epsilon^{(n)}(P)  \right] = 1.
\end{equation}
\end{proposition}

\begin{proposition}\label{prop:Stein}
For any $Q\in\PX$,
\begin{equation}
\lim_{\epsilon\downarrow 0} \lim_{n\to\infty} \frac{1}{n} \log \bigl( Q^{\times n} \left[ X^n \in \set{T}_\epsilon^{(n)}(P)\right]\bigr) = - D(P\|Q).
\end{equation}
\end{proposition}

We shall use Sanov's theorem in the following form where equality holds; for the more general form of Sanov's theorem, see~\cite{dembozeitouni98}.
\begin{proposition}[Sanov] \label{prop:Sanov}
Let $Q\in\set{P}(\set{X})$ have full support, i.e., $\supp(Q)=\set{X}$, and let $\set{S}\subseteq \set{P}(\set{X})$ be contained in the closure of its interior with respect to the subspace topology on $\set{P}(\set{X})$. 
Then
\begin{equation}
\lim_{n\to\infty} \frac{1}{n} \log \left( Q^{\times n} \big[ \tau[X^n] \in \set{S} \big] \right) = - \inf_{P\in\set{S}} D(P\| Q).
\end{equation}
\end{proposition}

\section{Typical Ensembles}\label{sec:typicalcoll}
Consider an ensemble of $\lceil 2^{nR} \rceil$ sequences
\begin{equation} \label{eq:ensemble}
 \set{E} = \left\{x^n(1),\ldots,x^n(\lceil 2^{nR} \rceil)\right\}, 
\end{equation}
where every $x^n(j)\in\set{X}^n$. The goal of this section is to introduce a notion of typicality for such ensembles. 

For a given distribution $P\in\PX$, fix some small $\delb>0$ and let
\begin{equation}\label{eq:defPRepsilon}
\set{P}_{R+\delb} \triangleq \{P'\ll P \colon D(P'\| P) \le R+\delb \}.
\end{equation}
Further let $\{\set{S}_a\}$, $a\in\set{A}$, be a finite partition of $\set{P}_{R+\delb}$ satisfying, for every $a\in\set{A}$,
\begin{IEEEeqnarray}{c}
\| P_1 - P_2 \|_1 \le \delb,\quad \textnormal{for all }P_1,P_2\in\set{S}_a \label{eq:defSa} \\
\set{S}_a  \subseteq \textnormal{clo} \bigl( \textnormal{int} \left( \set{S}_a  \right)\bigr).
\end{IEEEeqnarray}
A partition with such properties can be obtained, for example, by intersecting $\set{P}_{R+\delb}$ with a cubical grid on $\mathbb{R}^{|\set{X}|}$ of an appropriately chosen unit. Alternatively, one could pick a finite number of $\epsilon/2$-balls on  the probability simplex on $\supp(P)$ that cover $\set{P}_{R+\delb}$---which must be possible because $\set{P}_{R+\delb}$ is compact---and define the partition as the Voronoi diagram associated with the centers of the balls.

Denote, for every $a\in\set{A}$,
\begin{equation}\label{eq:defda}
d_a \triangleq \inf_{P'\in\set{S}_a} D(P'\| P).
\end{equation}
Let 
\begin{equation}
P_a^{(n)} \triangleq P^{\times n} \big[ \tau[X^n] \in \set{S}_a \big] = \sum_{x^n\colon \tau[x^n]\in\set{S}_a} P^{\times n}(x^n).
\end{equation}
By Sanov's theorem (Proposition~\ref{prop:Sanov}),
\begin{equation}\label{eq:Sanovda}
\lim_{n\to\infty} \frac{1}{n} \log P_a^{(n)} = - d_a.
\end{equation}


Suppose all sequences in $\mathcal{E}$ are IID $\sim P$ and independent of each other. Denote, for every $a\in\set{A}$,
\begin{equation}
N_a \triangleq \sum_{j=1}^{\lceil 2^{nR} \rceil} \boldsymbol{1}\left\{\tau\big[X^n(j)\big] \in \set{S}_a\right\}.
\end{equation}
Since $N_a$ is a sum of IID Bernoulli random variables of parameter $P_a^{(n)}$, we have
\begin{IEEEeqnarray}{rCl}
\E{N_a} & = & P_a^{(n)} \cdot \lceil 2^{nR} \rceil\\
\Var{N_a} & = & P_a^{(n)} (1- P_a^{(n)}) \cdot \lceil 2^{nR} \rceil.
\end{IEEEeqnarray}
It then follows from \eqref{eq:Sanovda} that
\begin{equation} \label{eq:prob1}
\lim_{n\to\infty} \Pr \bigl[2^{-n(d_a+\delb)} \le N_a 2^{-nR} \le 2^{-n(d_a - \delb)}\bigr] = 1.
\end{equation}

Motivated by the above observation, we now introduce the notion of a typical ensemble:

\begin{definition}\label{def:typens}
Given an ensemble $\set{E}$ of the form \eqref{eq:ensemble}, denote
\begin{equation}
\nu_a (\set{E}) \triangleq \sum_{j=1}^{\lceil 2^{nR} \rceil} \boldsymbol{1}\left\{\tau\big[x^n(j)\big] \in \set{S}_a\right\}.
\end{equation}
We say that $\set{E}$ is $\delb$-typical with respect to distribution $P$ and partition $\{\set{S}_a\}$, $a\in\set{A}$, if, for all $a\in\set{A}$,
\begin{equation}
2^{-n(d_a+\delb)} \le \nu_a(\set{E}) 2^{-nR} \le 2^{-n(d_a - \delb)},
\end{equation}
and if no $x^n(j)\in\set{E}$ is in $\set{P}_{R+\delb}^\textnormal{c}$. The set of all typical ensembles is denoted $\set{L}_\epsilon^{(n)}(P,\{\set{S}_a\})$.
\end{definition}

\begin{proposition}\label{prop:typical}
As $n\to\infty$, the probability that a random ensemble of $\lceil 2^{nR} \rceil$ $n$-sequences generated IID $\sim P$ is $\delb$-typical with respect to $P$ and $\{\set{S}_a\}$, $a\in\set{A}$, tends to one:
\begin{equation}
\lim_{n\to\infty} P^{\times (n \lceil 2^{nR} \rceil)} \Big[ \{X^n(1),\ldots,X^n(\lceil 2^{nR}\rceil )\} \in \set{L}_\epsilon^{(n)}(P,\{\set{S}_a\}) \Big] = 1.
\end{equation}
\end{proposition}

\begin{IEEEproof}
By Sanov's theorem (Proposition~\ref{prop:Sanov}), 
\begin{equation}
\lim_{n\to\infty} \frac{1}{n} \log \left( P^{\times n} \big[ \tau[X^n] \in \set{P}_{R+\delb}^\textnormal{c} \big] \right) = - R-\delb.
\end{equation}
Since there are only $\lceil 2^{nR} \rceil$ sequences, the probability that there exists at least one sequence in $\set{P}_{R+\delb}^\textnormal{c}$ tends to zero like $2^{-n\delb}$ as $n$ tends to infinity. The claim now follows by recalling \eqref{eq:prob1}, noting that $|\set{A}|$ is finite (and does not grow with $n$), 
and the union bound.
\end{IEEEproof}

\section{Proof of Direct Part}\label{sec:direct}
We describe a decision rule and analyze its error probabilities under both hypotheses. In the following, $\tau_j$ denotes the type of $x^n(j)$, $j\in\{1,\ldots,\lceil 2^{nR} \rceil\}$. 

\paragraph{Decision rule.} Fix some $\epsilon>0$. If there exists $\ell \in \{1,\ldots,\lceil 2^{nR} \rceil\}$ such that 
\begin{equation}\label{eq:elltyp}
x^n(\ell) \in \set{T}_\epsilon^{(n)}(Q_0),
\end{equation}
and if 
\begin{equation}
D\left( \tau_j \middle\| P \right) \le R + \delb\quad \textnormal{for all }j\neq\ell, \label{eq:violate}
\end{equation}
then declare $\mH_0$. In all other cases declare $\mH_1$.

\bigskip

In the analyses below, without loss of generality, we shall assume that 
\begin{equation}
K=1.
\end{equation}

\paragraph{Error under $\mH_0$.} Note that $X^n(1)$ is IID $\sim Q_0$. By Proposition~\ref{prop:LLN}, the probability that $X^n(1)\in \set{T}_\epsilon^{(n)}(Q_0)$ tends to one as $n\to\infty$. Furthermore, Proposition~\ref{prop:typical} implies that\footnote{For the direct part we only need a (much) weaker version of Proposition~\ref{prop:typical}: we only need the fact that the probability that there exists a sequence whose type is in $\set{P}_{R+\epsilon}^\textnormal{c}$ tends to zero as $n\to\infty$; i.e., the partition $\{S_a\}$ plays no role here. } the probability that there exists $j\neq 1$ that violates \eqref{eq:violate} tends to zero as $n\to\infty$. We hence conclude that the overall decision error probability under $\mH_0$ tends to zero. A fortiori, for any $\bar{\alpha}>0$, for sufficiently large $n$, this error probability is less than $\bar{\alpha}$.

\paragraph{Error under $\mH_1$.} There are two cases where an error will occur.

\textbf{Case 1:} $X^n(1)\in  \set{T}_\epsilon^{(n)}(Q_0)$ and \eqref{eq:violate} is satisfied. The latter happens with high probability (same as under $\mH_0$). As for the former, since $X^n(1)$ is IID $\sim Q_1$, by Proposition~\ref{prop:Stein}, when we make $\epsilon$ approach zero, the exponent of this probability approaches $D(Q_0\|Q_1)=\theta_1$. 

\textbf{Case 2:} there exists $\ell\neq 1$ such that 
\begin{equation}\label{eq:31}
X^n(\ell) \in \set{T}_\epsilon^{(n)}(Q_0),
\end{equation}
while 
\begin{equation}\label{eq:32}
D\big( \tau_1\big\| P \big) \le R + \delb.
\end{equation} 
Since $X^n(\ell)$ is IID $\sim P$, the probability that \eqref{eq:31} happens for a specific $\ell$ has exponent $D(Q_0\|P)$ (again when we make $\epsilon$ approach zero).\footnote{If $Q_0\not\ll P$, then \eqref{eq:31} can never happen, therefore the probability of Case~2 is zero. This is consistent with our claim: $\theta_2=\infty$ because $D(Q_0\|P)=\infty$.} Since there are $\lceil 2^{nR} \rceil-1$ mutually independent candidate sequences, the probability that there exists at least one such $\ell$ has exponent
\begin{equation}\label{eq:exp1}
\big( D(Q_0\|P) - R \big)^+.
\end{equation}
On the other hand, since $X^n(1)$ is IID $\sim Q_1$, by Sanov's theorem (Proposition~\ref{prop:Sanov}), the exponent of the probability for \eqref{eq:32} to happen is given by
\begin{equation}\label{eq:exp2}
\min_{P'\colon D(P'\|P) \le R+\delb} D(P'\| Q_1),
\end{equation}
where $\inf$ becomes $\min$ because the constraint set is closed and compact. 
Because $X^n(1)$ is independent of all other sequences in the ensemble, the events \eqref{eq:31} and \eqref{eq:32} are independent. Consequently, the overall exponent of the probability of Case~2 is the sum of \eqref{eq:exp1} and \eqref{eq:exp2}. When we make $\epsilon$ approach zero, this sum approaches $\theta_2$ as defined in \eqref{eq:theta2}. Here we use the fact
\begin{equation}
\lim_{\delb\downarrow 0}\, \min_{P'\colon D(P'\|P) \le R+\delb} D(P'\| Q_1) = \min_{P'\colon D(P'\|P) \le R} D(P'\| Q_1)
\end{equation}
because relative entropy is lower semicontinuous.

The overall error exponent under $\mH_1$ is dominated by the smaller exponent (i.e., the slower decay) between Case~1 and Case~2. It is therefore given by $\min \{\theta_1,\theta_2\}$. This concludes the proof of the direct part of Theorem~\ref{thm:main}.

\section{Proof of Converse Part}\label{sec:converse}

Even if the value of $K$ is known, the error exponent we seek cannot exceed the Chernoff-Stein exponent in testing between $Q_0$ and $Q_1$, which is $\theta_1$. We hence only need to show that (without the knowledge of $K$) the error exponent also cannot exceed $\theta_2$. We henceforth assume
\begin{equation}
Q_0\ll P
\end{equation}
because otherwise $\theta_2=\infty$ and the assertion holds trivially.

First consider the case where
\begin{equation}\label{eq:DgreaterR}
D(Q_0\| P) >R.
\end{equation}
Fix a small positive $\epsilon$, let $\set{P}_{R+\epsilon}$ be given by \eqref{eq:defPRepsilon}, and also fix a partition $\{\set{S}_a\}$, $a\in\set{A}$, of $\set{P}_{R+\epsilon}$ as described in Section~\ref{sec:typicalcoll}. 
For $i\in\{1,\ldots,\lceil 2^{nR} \rceil\}$, let $\set{B}_i$ denote the set of all $\{x^n(1),\ldots,x^n(\lceil 2^{nR} \rceil)\}$ satisfying\footnote{We defined $\set{L}_\epsilon^{(n)}(P,\{\set{S}_a\})$ for ensembles of $\lceil 2^{nR} \rceil$ sequences back in Section~\ref{sec:typicalcoll}, but here apply it to those of $(\lceil 2^{nR} \rceil-1)$ sequences. This difference is negligible for large $n$.}
\begin{IEEEeqnarray}{rCl}
x^n(i) & \in & \set{T}_\epsilon^{(n)}(Q_0)\\
\{x^n(j)\}_{j\neq i} & \in & \set{L}_\epsilon^{(n)}(P,\{\set{S}_a\}).
\end{IEEEeqnarray}
Under hypothesis $\mH_0$ and conditional on $K=i$, the probability that $\{X^n(1),\ldots,X^n(\lceil 2^{nR} \rceil)\}\in \set{B}_i$ tends to one as $n\to\infty$. It then follows that
\begin{equation}\label{eq:PUB}
\lim_{n\to\infty} \Pr \left[ \left\{X^n(1),\ldots,X^n(\lceil 2^{nR} \rceil)\right\}\in \bigcup_{i=1}^{\lceil 2^{nR} \rceil} \set{B}_i \,\middle|\, \mH_0\textnormal{ true}\right] = 1.
\end{equation}

We take any ensemble in the above union and compare its probability under both hypotheses. Without loss of generality, we can pick one in $\set{B}_1$: 
\begin{equation}\label{eq:inB1}
\{x^n(1),\ldots,x^n(\lceil 2^{nR} \rceil)\}\in\set{B}_1.
\end{equation} 
For every $j\in\{1,\ldots,\lceil 2^{nR} \rceil\}$, let $\tau_j$ denote the type of $x^n(j)$.
Denote
\begin{IEEEeqnarray}{rCl}
\omega & \triangleq & P^{\times \left(n(\lceil 2^{nR} \rceil-1)\right)} \Big[\left\{x^n(2),\ldots,x^n(\lceil 2^{nR} \rceil)\right\}\Big] \\
& = & \prod_{j=2}^{\lceil 2^{nR} \rceil} 2^{-n \left( H(\tau_j) + D(\tau_j\| P)\right)},
\end{IEEEeqnarray}
where the second equality follows by Proposition~\ref{prop:prob}.
Also denote, for every $a\in\set{A}$,
\begin{equation}
\nu_a \triangleq \sum_{j=2}^{\lceil 2^{nR} \rceil} \mathbf{1}\left\{ \tau_j \in\set{S}_a\right\}.
\end{equation}
Recalling Definition~\ref{def:typens}, \eqref{eq:inB1} implies
\begin{equation}\label{eq:boundnua}
2^{n(R - d_a - \epsilon)} \le \nu_a \le 2^{n(R-d_a+\epsilon)}.
\end{equation}

Under $\mH_0$ and conditional on $K=1$, the probability of $\{x^n(1),\ldots,x^n(\lceil 2^{nR} \rceil)\}$ is $Q_0^{\times n}(x^n(1))$ times $\omega$, i.e., it is
\begin{equation}
2^{-n(H(\tau_1) + D(\tau_1\|Q_0))}\cdot \omega. 
\end{equation}
Under $\mH_0$ and conditional on $K=j\neq1$, its probability is 
\begin{equation}
2^{-n(H(\tau_1)+D(\tau_1\|P))} \cdot \omega \cdot \frac{2^{-n(H(\tau_j) + D(\tau_j\|Q_0))}}{2^{-n(H(\tau_j)+D(\tau_j\|P))}},
\end{equation}
i.e., it is $P^{\times n}(x^n(1))$ times $\omega$ but with the term $P^{\times n}(x^n(j))$ in $\omega$ replaced by $Q_0^{\times n}(x^n(j))$. Since $K$ is uniform, we have
\begin{IEEEeqnarray}{rCl}
\IEEEeqnarraymulticol{3}{l}{
\Pr \left[ \left\{x^n(1),\ldots,x^n(\lceil 2^{nR} \rceil)\right\} \,\middle|\, \mH_0\textnormal{ true}\right] 
}\nonumber\\*
\qquad & = & 2^{-nR} \left( 2^{-n \left(H(\tau_1) + D(\tau_1\| Q_0) \right)} \cdot \omega +  \sum_{j=2}^{\lceil 2^{nR} \rceil} 2^{-n \left(H(\tau_1) + D(\tau_1\| P) \right)} \cdot \omega \cdot \frac{2^{-n(H(\tau_j) + D(\tau_j\|Q_0))}}{2^{-n(H(\tau_j)+D(\tau_j\|P))}} \right)\IEEEeqnarraynumspace \\
& = & 2^{-nR -nH(\tau_1)} \cdot \omega \cdot\left( 2^{-nD(\tau_1\|Q_0)} + 2^{-n D(\tau_1\| P)} \sum_{j=2}^{\lceil 2^{nR} \rceil} 2^{-n\left(D(\tau_j \| Q_0) - D(\tau_j\| P)\right)} \right)\\
& \le& 2^{-nR-n H(\tau_1)}\cdot \omega \cdot \left( 1 + 2^{-n D(\tau_1\| P)} \sum_{j=2}^{\lceil 2^{nR} \rceil} 2^{n D(\tau_j\| P)} \right)\\
& = & 2^{-nR-n H(\tau_1)}\cdot \omega \cdot \left( 1 + 2^{-n D(\tau_1\| P)} \sum_{a\in\set{A}} \sum_{j=2}^{\lceil 2^{nR} \rceil} \mathbf{1}\{\tau_j \in\set{S}_a\} 2^{n D(\tau_j\| P)}  \right). \label{eq:54}
\end{IEEEeqnarray}
Here we need to use the fact that $D(P'\|P)$ is uniformly continuous in $P'\ll P$ with respect to total variation distance. Namely, there exists some $\lambda(\epsilon)$---which may depend on $P$---that tends to zero as $\epsilon$ tends to zero such that, if $P_1,P_2\ll P$ and $\| P_1-P_2\|_1 \le \epsilon$, then 
\begin{IEEEeqnarray}{rCl}
|D(P_1\|P)-D(P_2\|P)| 
& \le & \lambda(\epsilon). \label{eq:lameps}
\end{IEEEeqnarray}
This is because convergence in total variation implies component-wise convergence; the mapping $\xi\mapsto \xi\log(\xi/\alpha)$ is continuous whenenver $\alpha>0$; and $P$ has finite support hence the minimum of $P(x)>0$ exists. 
With \eqref{eq:lameps}, we can continue \eqref{eq:54} as
\begin{IEEEeqnarray}{rCl}
\IEEEeqnarraymulticol{3}{l}{
\Pr \left[ \left\{x^n(1),\ldots,x^n(\lceil 2^{nR} \rceil)\right\} \,\middle|\, \mH_0\textnormal{ true}\right] 
}\nonumber\\*
\qquad\qquad& \le & 2^{-nR-n H(\tau_1)}\cdot \omega \cdot \left( 1 + 2^{-n D(\tau_1\| P)} \sum_{a\in\set{A}} \nu_a 2^{n \left(d_a + \lambda( \epsilon)\right)} \right) \label{eq:51da} \\
& \le & 2^{-nR-n H(\tau_1)}\cdot \omega \cdot \bigg( 1 + 2^{-n D(\tau_1\| P)} \left| \set{A} \right| \cdot 2^{n\left(R+\lambda(\epsilon)+\epsilon\right)} \bigg), \label{eq:50Rlambda}
\end{IEEEeqnarray}
where \eqref{eq:51da} follows from the definition of $d_a$ in \eqref{eq:defda} and \eqref{eq:lameps}; and \eqref{eq:52Rlambda} follows by \eqref{eq:boundnua}.
Recall that $x^n(1)\in\set{T}_\epsilon^{(n)}$, which implies that $\|\tau_1-Q_0\|_1\le \epsilon$. Consequently, again by \eqref{eq:lameps},
\begin{equation}\label{eq:DtauQP}
\left| D(\tau_1\| P) - D(Q_0\| P) \right| \le \lambda (\epsilon),
\end{equation}
so \eqref{eq:50Rlambda} can be relaxted to
\begin{IEEEeqnarray}{rCl}
\IEEEeqnarraymulticol{3}{l}{
\Pr \left[ \left\{x^n(1),\ldots,x^n(\lceil 2^{nR} \rceil)\right\} \,\middle|\, \mH_0\textnormal{ true}\right] 
}\nonumber\\*
\qquad\qquad& \le & 2^{-nR-n H(\tau_1)}\cdot \omega \cdot \left(1+ |\set{A}| \cdot 2^{n(R-D(Q_0\|P) + 2\lambda(\epsilon)+\epsilon)}\right).
\label{eq:52Rlambda}
\end{IEEEeqnarray}
Since we assume \eqref{eq:DgreaterR} to hold, there must exist a sufficiently small $\epsilon$ such that the second summand inside the parenthesis in \eqref{eq:52Rlambda} should decay to zero exponentially with $n$. Therefore, for any $\delta>0$, there exist $\epsilon>0$ and partition $\{\set{S}_a\}$ such that, for sufficiently large $n$, 
\begin{equation} \label{eq:PrH0}
\Pr \left[ \left\{x^n(1),\ldots,x^n(\lceil 2^{nR} \rceil)\right\} \,\middle|\, \mH_0\textnormal{ true}\right]  \le (1+\delta) 2^{-nR-n H(\tau_1)}\cdot \omega.
\end{equation}

We now analyze the probability of the same ensemble \eqref{eq:inB1} under $\mH_1$:
\begin{IEEEeqnarray}{rCl}
\IEEEeqnarraymulticol{3}{l}{
\Pr \left[ \left\{x^n(1),\ldots,x^n(\lceil 2^{nR} \rceil)\right\} \,\middle|\, \mH_1\textnormal{ true}\right] 
}\nonumber\\*
\qquad & = &2^{-nR}\left( 2^{-n(H(\tau_1)+ D(\tau_1\|Q_1))} \cdot \omega + \sum_{j=2}^{\lceil 2^{nR} \rceil} 2^{-n(H(\tau_1) + D(\tau_1\|P))} \cdot \omega \cdot \frac{2^{-n(H(\tau_j) + D(\tau_j\|Q_1))}}{2^{-n(H(\tau_j)+D(\tau_j\|P))}}\right) \IEEEeqnarraynumspace \\
& \ge & 2^{-nR-nH(\tau_1)}\cdot\omega \cdot 2^{-nD(\tau_1\| P)} \sum_{j=2}^{\lceil 2^{nR} \rceil} 2^{n(D(\tau_j\| P ) - D(\tau_j\| Q_1))} \\
& = & 2^{-nR-nH(\tau_1)}\cdot\omega \cdot 2^{-nD(\tau_1\| P)} \sum_{a\in\set{A}} \sum_{j=2}^{\lceil 2^{nR} \rceil} \mathbf{1}\{\tau_j\in\set{S}_a\} 2^{n(D(\tau_j\| P ) - D(\tau_j\| Q_1))}. \label{eq:57} \IEEEeqnarraynumspace
\end{IEEEeqnarray}
Now we again use continuity of relative entropy: there exists some $\hat{\lambda}(\epsilon)$---which may depend on $Q_1$---that tends to zero as $\epsilon$ tends to zero such that, for $P_1,P_2\ll Q_1$ (which is true for all elements of $\set{P}_{R+\epsilon}$ because we assume $P\ll Q_1$) satisfying $\|P_1-P_2\|_1 \le \epsilon$, 
\begin{equation}
| D(P_1\| Q_1) - D(P_2\|Q_1)| \le \hat{\lambda}(\epsilon).
\end{equation}
Now define
\begin{equation}
r_a \triangleq \inf_{P'\in\set{S}_a} D(P'\|Q_1),
\end{equation}
then, for all $\tau\in\set{S}_a$,
\begin{equation}
D(\tau\|Q_1) \le r_a + \hat{\lambda} (\epsilon).
\end{equation}
Further recalling \eqref{eq:defda}, we can continue \eqref{eq:57} as
\begin{IEEEeqnarray}{rCl}
\IEEEeqnarraymulticol{3}{l}{
\Pr \left[ \left\{x^n(1),\ldots,x^n(\lceil 2^{nR} \rceil)\right\} \,\middle|\, \mH_1\textnormal{ true}\right] 
}\nonumber\\*
\qquad\qquad & \ge & 2^{-nR-nH(\tau_1)}\cdot\omega \cdot 2^{-nD(\tau_1\| P)} \sum_{a\in\set{A}} \nu_a 2^{n(d_a - r_a - \hat{\lambda} (\epsilon))}  \label{eq:defra}\\
& \ge & 2^{-nR-nH(\tau_1)}\cdot\omega \cdot 2^{-nD(\tau_1\| P)} \sum_{a\in\set{A}} 2^{n(R-r_a-\epsilon - \hat{\lambda} (\epsilon))} \label{eq:58H1}\\
& \ge & 2^{-nR-nH(\tau_1)}\cdot\omega \cdot 2^{-n\bigl( D(\tau_1\| P) +  r^* -R +  \epsilon +\hat{\lambda}(\epsilon) \bigr)}.
\label{eq:defrstar} 
\end{IEEEeqnarray}
Here, \eqref{eq:58H1} follows by \eqref{eq:boundnua};
and in \eqref{eq:defrstar} we discarded all except the largest summand and introduced 
\begin{equation}
r^* \triangleq \min_{P'\in\set{P}_{R+\delb}} D(P'\|Q_1) = \min_{a\in\set{A}} r_a.
\end{equation}
Using \eqref{eq:DtauQP} again, we can further relax \eqref{eq:58H1} to
\begin{equation}
\Pr \left[ \left\{x^n(1),\ldots,x^n(\lceil 2^{nR} \rceil)\right\} \,\middle|\, \mH_1\textnormal{ true}\right]  \ge 2^{-nR-nH(\tau_1)}\cdot\omega \cdot 2^{-n\bigl( D(Q_0\| P) +  r^* -R +  \epsilon +\hat{\lambda}(\epsilon)+\lambda(\epsilon)\bigr)}. \label{eq:PrH1}
\end{equation}

Combining the two bounds \eqref{eq:PrH0} and \eqref{eq:PrH1} we now have, for every ensemble \eqref{eq:inB1},
the ratio between the probabilities under both hypotheses satisfies
\begin{equation}\label{eq:Pratio}
\frac{\Pr \left[ \big\{x^n(1),\ldots,x^n(\lceil 2^{nR} \rceil)\big\} \,\middle|\, \mH_1\textnormal{ true}\right]}{\Pr \left[ \big\{x^n(1),\ldots,x^n(\lceil 2^{nR} \rceil)\big\} \,\middle|\, \mH_0 \textnormal{ true}\right]} \ge (1+\delta)^{-1} \cdot 2^{-n\bigl( D(Q_0\| P) +  r^* -R + \epsilon +\hat{\lambda}(\epsilon) +\lambda(\epsilon)\bigr)}.
\end{equation}
Further note that this bound does not depend on the assumption that $i=1$. Namely, \eqref{eq:Pratio} holds for all
\begin{equation}
\{x^n(1),\ldots,x^n(\lceil 2^{nR} \rceil)\}\in\bigcup_{i=1}^{\lceil 2^{nR} \rceil} \set{B}_i.
\end{equation}

Suppose a sequence of decision rules indexed by $n$, whose decision regions for $\mH_0$ are $\{\set{D}_n\}$, has error probabilities bounded away from $1$, i.e., for sufficiently large $n$,
\begin{equation}
\Pr \left[ \left\{X^n(1),\ldots,X^n(\lceil 2^{nR} \rceil)\right\}\in \set{D}_n \,\middle|\, \mH_0\textnormal{ true}\right] \ge 1-\bar{\alpha},
\end{equation}
where $\bar{\alpha}\in(0,1)$. It then follows by \eqref{eq:PUB} that
\begin{equation} \label{eq:H067}
\liminf_{n\to\infty} \Pr \left[ \left\{X^n(1),\ldots,X^n(\lceil 2^{nR} \rceil)\right\}\in \set{D}_n \cap \left( \bigcup_{i=1}^{\lceil 2^{nR} \rceil} \set{B}_i \right) \,\middle|\, \mH_0\textnormal{ true}\right] \ge 1-\bar{\alpha}.
\end{equation}
The error probability of such a decision rule under $\mH_1$ is bounded as
\begin{IEEEeqnarray}{rCl}
\beta_n & = & \Pr\left[ \left\{X^n(1),\ldots,X^n(\lceil 2^{nR} \rceil)\right\} \in \set{D}_n \,\middle|\, \mH_1\textnormal{ true} \right]\\
& \ge & \Pr \left[ \left\{X^n(1),\ldots,X^n(\lceil 2^{nR} \rceil)\right\}\in \set{D}_n \cap \left( \bigcup_{i=1}^{\lceil 2^{nR} \rceil} \set{B}_i \right) \,\middle|\, \mH_1\textnormal{ true}\right]. \label{eq:H169}
\end{IEEEeqnarray}
We can now combine \eqref{eq:Pratio}, \eqref{eq:H067}, and \eqref{eq:H169} to obtain the following bound:
\begin{IEEEeqnarray}{rCl}
\liminf_{n\to\infty} \frac{1}{n}\log \beta_n & \ge & \lim_{n\to\infty} \frac{1}{n}\log\left(\frac{1-\bar{\alpha}}{1+\delta}\right) - \bigl( D(Q_0\| P) +  r^* -R +  \epsilon +\hat{\lambda}(\epsilon)+\lambda(\epsilon)\bigr)\\
& = & - \bigl( D(Q_0\| P) +  r^* -R +  \epsilon +\hat{\lambda}(\epsilon)+\lambda(\epsilon)\bigr). \label{eq:beta72}
\end{IEEEeqnarray}
When we let $\epsilon\downarrow 0$, $r^*$ tends to 
\begin{equation}
\lim_{\epsilon\downarrow 0} \min_{P\colon D(P'\| P) \le R+\epsilon } D(P'\|Q_1) = \min_{P'\colon D(P'\| P)\le R} D(P'\|Q_1).
\end{equation}
Consequently, the exponent of the right-hand side of \eqref{eq:beta72} tends to $-\theta_2$. We have shown that the exponential decay in $\beta_n$ cannot be faster than $\theta_2$ and hence completed the converse proof for the case where \eqref{eq:DgreaterR} holds.

It remains to prove the converse for the case where
\begin{equation}\label{eq:DleR}
D(Q_0\| P) \le R.
\end{equation}
In this case, \eqref{eq:50Rlambda} remains valid, but now it is dominated by the second summand: for any $\delta>0$, for sufficiently large $n$,
\begin{IEEEeqnarray}{rCl}
\Pr \left[ \left\{x^n(1),\ldots,x^n(\lceil 2^{nR} \rceil)\right\} \,\middle|\, \mH_0\textnormal{ true}\right]  & \le & (1+\delta) 2^{-nR-n H(\tau_1)}\cdot \omega \cdot |\set{A}|\cdot 2^{n\bigl(R-D(\tau_1\|P) + \epsilon + \lambda(\epsilon) \bigr)}. \IEEEeqnarraynumspace \label{eq:75new}
\end{IEEEeqnarray}
The bound \eqref{eq:defrstar} also remains valid. 
Combining \eqref{eq:75new} and \eqref{eq:defrstar}, we obtain
\begin{equation}\label{eq:Prationew}
\frac{\Pr \left[ \big\{x^n(1),\ldots,x^n(\lceil 2^{nR} \rceil)\big\} \,\middle|\, \mH_1\textnormal{ true}\right]}{\Pr \left[ \big\{x^n(1),\ldots,x^n(\lceil 2^{nR} \rceil)\big\} \,\middle|\, \mH_0 \textnormal{ true}\right]} \ge (1+\delta)^{-1} |\set{A}|^{-1} 2^{-n\left(r^* + 2\epsilon +\lambda(\epsilon)+\hat{\lambda}(\epsilon)\right)}.
\end{equation}
Combining \eqref{eq:Prationew}, \eqref{eq:H067}, and \eqref{eq:H169}, we obtain
\begin{IEEEeqnarray}{rCl}
\liminf_{n\to\infty} \frac{1}{n} \log \beta_n & \ge & \lim_{n\to\infty} \frac{1}{n} \log\frac{1}{(1+\delta)|\set{A}|} - r^* - \left(2\epsilon +\lambda(\epsilon)+\hat{\lambda}(\epsilon)\right)\\
& = & - r^* - \left(2\epsilon +\lambda(\epsilon)+\hat{\lambda}(\epsilon)\right).
\end{IEEEeqnarray}
When we let $\epsilon\downarrow 0$,  $r^*$ approaches 
\begin{equation}
\min_{P'\colon D(P'\| P)\le R} D(P'\|Q_1) = \theta_2,
\end{equation}
completing the proof. Notice that, whenever \eqref{eq:DleR} holds, $\theta_2\le \theta_1$ so the error exponent in the main formula \eqref{eq:main} equals~$\theta_2$.



\section{Final Remarks}\label{sec:conclusion}

We recall the connection between the problem studied above and that of distributed hypothesis testing. A random binning scheme for the latter problem was proposed in \cite{shimokawahanamari94} together with a decision rule; a better decision rule was later proposed in \cite{kochmanwang25}. The decision problem in \cite{shimokawahanamari94,kochmanwang25} is very similar to the one studied here, the main difference being that in \cite{shimokawahanamari94,kochmanwang25} the outlier sequence is \emph{conditionally} IID: its components are generated, given a fixed sequence $y^n$, independently according to conditional PMF $W_0(x|y)$ or $W_1(x|y)$ respectively under $\mH_0$ and $\mH_1$.  As noted earlier, the decision rule described by \eqref{eq:elltyp} and \eqref{eq:violate} is directly adapted from the one in \cite{kochmanwang25}; only the converse proof is new. Our converse proof readily extends to the conditionally-IID-outlier setting; we omit the details. Such a converse then implies that the decision rule in \cite{kochmanwang25} is optimal in its own setting, namely distributed hypothesis testing when the sender employs the random binning scheme of \cite{shimokawahanamari94}. 




\end{document}